\documentclass[a4paper,UKenglish]{mylipics-v2016}

\usepackage{microtype}

\title{Communicating Finite-State Machines and Two-Variable Logic}

\author[1]{Benedikt Bollig}
\author[2]{Marie Fortin}
\author[2]{Paul Gastin}
\affil[1]{CNRS, ENS Paris-Saclay, LSV, Universit{\'e} Paris-Saclay}
\affil[2]{ENS Paris-Saclay, CNRS, LSV, Universit{\'e} Paris-Saclay}

\authorrunning{B. Bollig, M. Fortin, and P. Gastin}

\Copyright{Benedikt Bollig, Marie Fortin, and Paul Gastin}

\makeatletter
\renewcommand{\paragraph}{\@startsection{paragraph}{6}{\z@}{2ex}{-0.7em}{\normalsize\bf}}
\makeatother

\usepackage{amsmath,amssymb}
\usepackage{stmaryrd}
\usepackage{fancyhdr}
\usepackage{xspace}
\usepackage{mathtools}
\usepackage{cite}
\usepackage{tikz}
\usetikzlibrary{trees,arrows,shapes,snakes,fit,shadows,decorations.text,decorations.pathmorphing,decorations.pathreplacing,calc,automata,positioning,patterns}

\usepackage{ifpdf}
\ifpdf
\usepackage[framemethod=TikZ]{mdframed}
\usepackage[colorinlistoftodos,prependcaption]{todonotes}

\fi

\definecolor{mygreen}{RGB}{0.0,180,0.0}

\newcounter{todocounter}

\usepackage{amsmath,amssymb}
\usepackage{stmaryrd}
\usepackage{xspace}
\usepackage{mathtools}
\usepackage{tikz}
\usetikzlibrary{trees,arrows,shapes,snakes,fit,shadows,decorations.text,decorations.pathmorphing,decorations.pathreplacing,calc,automata,positioning,patterns}

\tikzstyle{dot} = [circle, fill, inner sep=0, minimum size = 2pt]
\tikzstyle{bdot} = [circle, fill, inner sep=0, minimum size = 3pt]

\definecolor{mygreen}{RGB}{0.0,180,0.0}

\newcommand\A{\ensuremath{\mathcal{A}}\xspace}

\newcommand\prel{\rightarrow}
\newcommand\mrel{\lhd}
\newcommand\init{\iota}

\newcommand{\nproc}{u}

\newcommand\EMSOt{\ensuremath{\textup{EMSO}^2}\xspace}
\newcommand\MSO{\ensuremath{\textup{MSO}}\xspace}
\newcommand\EMSO{\ensuremath{\textup{EMSO}}\xspace}
\newcommand\FOt{\ensuremath{\textup{FO}^2}\xspace}
\newcommand\FO{\ensuremath{\textup{FO}}\xspace}

\newcommand\CFM{\text{CFM}\xspace}
\newcommand\CFMs{\text{CFMs}\xspace}
\newcommand\cfm{\text{communicating finite-state machine}\xspace}
\newcommand\cfms{\text{communicating finite-state machines}\xspace}

\newcommand\Cfms{\text{Communicating finite-state machines}\xspace}
\newcommand\Acc{\mathit{Acc}}

\newcommand\Parallelp[2]{\mathsf{Parallel}_{#1}(#2)}
\newcommand\rPastp[2]{{\downarrow_{#1}} (#2)}
\newcommand\rFuturep[2]{{\uparrow_{#1}} (#2)}
\newcommand\type{\mathsf{type}}
\newcommand\Atype{\ensuremath{\A_{\mathsf{types}}}\xspace}

\newcommand\ttype{\tau}

\newcommand\p{\pi}
\newcommand\action{\alpha}
\newcommand\source{\mathit{source}}
\newcommand\target{\mathit{target}}
\newcommand\tmsg{\mathit{msg}}
\newcommand\tlabel{\mathit{label}}
\newcommand\receiver{\mathit{receiver}}
\newcommand\sender{\mathit{sender}}
\newcommand\MSCs[2]{\mathbb{MSC}(#1,#2)}
\newcommand\trans{t}
\newcommand\Paths{\Pi}
\newcommand\last[2]{\mathsf{pred}_{#1}(#2)}
\newcommand\first[2]{\mathsf{succ}_{#1}(#2)}

\newcommand\signature{R}
\newcommand\Procs{P}
\newcommand\Ch{\mathit{Ch}}

\newcommand\Msg{\mathit{Msg}}
\newcommand\msg{\textup{m}}
\newcommand\labloc{\lambda}

\newcommand\Lzeros{L_0}
\newcommand\Lones{L_1}
\newcommand\Azeros{\A_0}
\newcommand\Aones{\A_1}
\newcommand\Lo{L_{\mathsf{parallel}}}
\newcommand\Lint{L_{\mathsf{intervals}}}
\newcommand\Lleft{L_{\mathsf{left}}}
\newcommand\Lright{L_{\mathsf{right}}}
\newcommand\Aint{\A_{\mathsf{intervals}}}
\newcommand\Ao{\A_{\mathsf{parallel}}}
\newcommand\Aleft{\A_{\mathsf{left}}}
\newcommand\Aright{\A_{\mathsf{right}}}

\newcommand\addlab{\gamma}
\newcommand\Rel{\Omega}
\newcommand\strictless{\ll}
\newcommand\apprel[3]{\mathcal{E}_{#1}({#2},#3)}
\newcommand\apprelp[2]{\mathcal{E}({#1},#2)}
\newcommand\Types{\mathbb{T}_{\Procs,\extSigma}}

\newcommand\extSigma{\Sigma'}
\newcommand\nvar{m}
\newcommand\nform{\ell}

\newcommand\rpast[1]{{\downarrow}#1}
\newcommand\past[1]{{\Downarrow}#1}
\newcommand\rpasttype[1]{\lambda(\rpast{#1})}
\newcommand\pasttype[1]{\lambda(\past{#1})}

\newcommand{\conc}{\parallel}

\newcommand\chemin{\nu}

\newcommand\xqed[1]{%
  \leavevmode\unskip\penalty9999 \hbox{}\nobreak\hfill
  \quad\hbox{#1}}
\newcommand\eofex{\xqed{$\lhd$}}

\begin{document}

\maketitle

\pagestyle{fancy}
\fancyhead{}
\renewcommand{\headrulewidth}{0pt}
\fancyfoot[C]{\vspace{2ex}\thepage}



\begin{abstract}
\Cfms are a fundamental, well-studied model of finite-state processes that communicate via unbounded first-in first-out channels. We show that they are expressively equivalent to existential MSO logic with two first-order variables and the order relation.
\end{abstract}

\section{Introduction}

The study of logic-automata connections has ever played a key role in computer 
science, relating concepts that are a priori very different. Its motivation is at least twofold.
First, automata may serve as a tool to decide logical theories. Beginning with the work of B{\"u}chi, Elgot, and Trakhtenbrot, who established expressive equivalence of monadic second-order (MSO) logic and finite automata \cite{Buechi:60,Elgot1961,Trakhtenbrot62}, the ``automata-theoretic'' approach to logic has been successfully applied, for example, to MSO logic on trees \cite{ThaWri68}, temporal logics \cite{VardiW86}, and first-order logic with two variables over words with an equivalence relation (aka \emph{data words}) \cite{Bojanczy06}.
Second, automata serve as models of various kind of state-based systems. Against this background, B{\"u}chi-like theorems lay the foundation of \emph{synthesis}, i.e., the process of transforming high-level specifications (represented as logic formulas) into faithful system models. In this paper, we provide a B{\"u}chi theorem for \emph{\cfms}, which are a classical model of concurrent message-passing systems.

\smallskip

One of the simplest system models are finite automata. They can be considered 
as single finite-state processes and, therefore, serve as a model of 
\emph{sequential} systems. Their executions are words, which, seen as a logical structure, consist of a set of positions (also referred to as \emph{events}) that carry letters from a finite alphabet and are \emph{linearly} ordered by some binary relation $\le$. The simple MSO (even first-order) formula $\forall x. \bigl(a(x) \Longrightarrow \exists y.(x \le y \wedge b(y))\bigr)$ says that every ``request'' $a$ is eventually followed by an ``acknowledgment'' $b$. In fact, B{\"u}chi's theorem allows one to turn any logical MSO specification into a finite automaton. The latter can then be considered correct by construction.
Though the situation quickly becomes more intricate when we turn to other automata models, B{\"uchi} theorems have been established for expressive generalizations of finite automata that also constitute natural system models. In the following, we will discuss some of them.

\emph{Data automata} accept (in the context of system models, we may also say \emph{generate}) words that, in addition to the linear order $\le$ and its direct-successor relation, are equipped with an equivalence relation $\sim$ \cite{Bojanczy06}. Positions (events) that belong to the same equivalence class may be considered as being executed by one and the same process, while $\le$ reflects a sort of global control. It is, therefore, convenient to also include a predicate that connects \emph{successive} events in an equivalence class. Boja{\'n}czyk et al.\ showed that data automata are expressively equivalent to existential MSO logic with two first-order variables \cite{Bojanczy06}. A typical formula is $\neg\exists x. \exists y. (x \neq y \wedge x \sim y)$, which says that every equivalence class is a singleton. It should be noted that data automata scan a word twice and, therefore, can hardly be seen as a system model. However, they are expressively equivalent to \emph{class-memory automata}, which distinguish between a global control (modeling, e.g., a shared variable) and a local control for every process \cite{Bjorklund10}.

Unlike finite automata and data automata, \emph{asynchronous automata} are a model of \emph{concurrent} shared-memory systems, with a finite number of processes. Their executions are Mazurkiewicz traces, where the relation $\le$ is no longer a total, but a partial order. Thus, there may be \emph{parallel} events $x$ and $y$, for which neither $x \le y$ nor $y \le x$ holds. A typical logical specification is the mutual exclusion property, which can be expressed in MSO logic as $\neg \exists x. \exists y. (\mathit{CS}(x) \wedge \mathit{CS}(y) \wedge x \conc y)$ where the parallel operator $x \conc y$ is defined as $\neg (x \le y) \wedge \neg (y \le x)$. Note that this is even a first-order formula that uses only two first-order variables, $x$ and $y$. It says that there are no two events $x$ and $y$ that access a critical section simultaneously. Asynchronous automata are closed under complementation \cite{Zielonka87} so that the inductive approach to translating formulas into automata can be applied to obtain a B{\"u}chi theorem \cite{tho90traces}. Note that complementability is also the key ingredient for MSO characterizations of \emph{nested-word automata} \cite{Alur2009} and \emph{branching automata} running over series-parallel posets (aka N-free posets) \cite{Kuske00,Bedon15}.

The situation is quite different in the realm of \emph{\cfms (\CFMs)}, aka
\emph{communicating automata} or \emph{message-passing automata}, where finitely
many processes communicate by exchanging messages through \emph{unbounded} FIFO
channels \cite{Brand1983}.  A \CFM accepts/generates message-sequence charts
(MSCs) which are also equipped with a partial order $\le$.  Additional binary
predicates connect (i) the emission of a message with its reception, and (ii)
successive events executed by one and the same process.  Unfortunately, \CFMs
are not closed under complementation \cite{BolligJournal} so that an inductive
translation of MSO logic into automata will fail.  In fact, they are strictly
less expressive than MSO logic.  Two approaches have been adopted to overcome
these problems.  First, when channels are (existentially or universally)
bounded, closure under complementation is recovered so that \CFMs are
expressively equivalent to MSO logic
\cite{HenriksenJournal,GKM06,Kuske01,GKM07}.  Note that, however, the
corresponding proofs are much more intricate than in the case of finite
automata.  Second, \CFMs with unbounded channels have been shown
to be expressively equivalent to \emph{existential} MSO logic when dropping the
order $\le$ \cite{BolligJournal}.  The proof relies on Hanf's normal form of
first-order formulas on structures of \emph{bounded degree} (which is why one
has to discard $\le$) \cite{Hanf1965}. However, it is clear that many specifications
(such as mutual exclusion) are easier to express in terms of $\le$.
But, to the best of our knowledge, a convenient specification language that is exactly as expressive as \CFMs has still
been missing.

\smallskip

It is the aim of this paper to close this gap, i.e., to provide a logic that
\begin{itemize}
\item \emph{matches} exactly the expressive power of \emph{unrestricted} \CFMs (in particular, every specification should be \emph{realizable} as an automaton), and
\item \emph{includes} the order $\le$ so that one can easily express natural properties like mutual exclusion.
\end{itemize}

We show that existential MSO logic with two first-order variables is an appropriate logic. To translate a formula into an automaton, we first follow the approach of \cite{Bojanczy06} for data automata and consider its Scott normal form (cf.\ \cite{GradelO99}). However, while data automata generate total orders, the main difficulty in our proof comes from the fact that $\le$ is a partial order. Actually, our main technical contribution is a \CFM that, running on an MSC, marks precisely those events that are in parallel to some event of a certain type.

\paragraph{Outline.} The paper is structured as follows. In Section~\ref{sec:prel}, we recall the classical notions of \CFMs and MSO logic. Section~\ref{sec:result} states our main result, describes our proof strategy, and settles several preliminary lemmas. The main technical part is contained in Section~\ref{sec:parallel}. We conclude in Section~\ref{sec:conclusion}.



\section{Preliminaries}\label{sec:prel}

Let $\Sigma$ be a finite alphabet. The set of finite words over $\Sigma$ is denoted by $\Sigma^\ast$, which includes the empty word $\epsilon$. For $w \in \Sigma^\ast$, let $|w|$ denote its length. In particular, $|\epsilon| = 0$. The inverse of a binary relation $R$ is defined as ${R}^{-1} = \{(f,e) \mid (e,f) \in {R}\}$.
We denote the size of a finite set $A$ by $|A|$.

\subsection{Communicating Finite-State Machines}

\Cfms are a natural model of communicating systems where a finite number of processes communicate through a priori unbounded FIFO channels \cite{Brand1983}. Every process is represented as a \emph{finite transition system} $(S,\init,\Delta)$ over some finite alphabet $\Gamma$, i.e., $S$ is a finite set of states with initial state $\init \in S$, and $\Delta \subseteq S \times \Gamma \times S$ is the transition relation. Elements from $\Gamma$ will describe the action that is performed when taking a transition (e.g., ``send a message to some process'' or ``perform a local computation'').

A \cfm is a \emph{collection} of finite transition systems, one for each process. For the rest of this paper, we fix a finite set $\Procs = \{p,q,r,\ldots\}$ of \emph{processes} and a finite alphabet $\Sigma=\{a,b,c,\ldots\}$ of \emph{labels}. We assume that there is a channel between any two distinct processes. Thus, the set of \emph{channels} is $\Ch = \{(p,q) \in \Procs \times \Procs \mid p \neq q\}$.

\begin{definition}
A \emph{\cfm (\CFM)} over $\Procs$ and $\Sigma$ is a tuple $\A=((\A_p)_{p \in \Procs},\Msg,\Acc)$ where
\begin{itemize}
\item $\Msg$ is a finite set of \emph{messages},

\item $\A_p = (S_p,\init_p,\Delta_p)$ is a finite transition system over $\Sigma \cup (\Sigma \times \{!\,,?\} \times \Msg \times (\Procs \setminus \{p\}))$, and

\item $\Acc \subseteq \prod_{p \in \Procs} S_p$ is the set of \emph{global accepting states}.%
\footnote{We may also include several \emph{global} initial states without changing the expressive power, which is convenient in several of the forthcoming constructions.} \eofex
\end{itemize}
\end{definition}

Let $\trans=(s,\action,s') \in \Delta_p$ be a transition of process $p$. We call $s$ the \emph{source state} of $\trans$, denoted by $\source(\trans)$, and $s'$ its \emph{target state}, denoted $\target(\trans)$. Moreover, $\action$ is the \emph{action} executed by~$t$. If $\action \in \Sigma$, then $t$ is said to be \emph{internal}, and we let $\tlabel(t) = \action$.
The label from $\Sigma$ may provide some more information about an event (such as ``enter critical section'').
When $\action$ is of the form $(a,!\,,\msg,q)$, then $t$ is a send transition, which writes message $\msg$ into the channel $(p,q)$. Accordingly, we let $\tmsg(\trans) = \msg$, $\receiver(\trans) = q$, and $\tlabel(t) = a$. Finally, performing $\action = (a,?,\msg,q)$ removes message $\msg$ from channel $(q,p)$. In that case, we set $\tmsg(\trans) = \msg$, $\sender(\trans) = q$, and $\tlabel(t) = a$.

If there is only one process, i.e., $\Procs$ is a singleton, then all
transitions are internal so that a \CFM is simply a finite automaton accepting a
regular set of words over the alphabet $\Sigma$.  In the presence of several
processes, a single behavior is a \emph{collection} of words over $\Sigma$, one
for every process.  However, these words are not completely independent (unless
all transitions are internal and there is no communication), since the sending
of a message can be linked to its reception.  This is naturally reflected by a
binary relation $\mrel$ that connects word positions on distinct processes.  The
resulting structure is called a message sequence chart.

\begin{definition}
A \emph{message sequence chart (MSC)} over $\Procs$ and $\Sigma$ is a tuple $M = ((w_p)_{p \in \Procs},\mrel)$ where $w_p \in \Sigma^\ast$ for every $p \in \Procs$. We require that at least one of these words be non-empty. By $E_p = \{p\} \times \{1,\ldots,|w_p|\}$, we denote the set of \emph{events} that are executed by process $p$. Accordingly, the (disjoint) union $E = \bigcup_{p \in \Procs} E_p$ is the set of all events.
Implicitly, we obtain the process-edge relation ${\prel} \subseteq \bigcup_{p \in \Procs} (E_p \times E_p)$, which connects successive events that are executed by one and the same process: $(p,i) \to (p,i+1)$ for all $p \in \Procs$ and $i \in \{1,\ldots,|w_p|-1\}$.
Now, ${\mrel} \subseteq \bigcup_{(p,q) \in \Ch} (E_p \times E_q)$ is a set of \emph{message edges}, satisfying the following:
\begin{itemize}
  \item $({\prel} \cup {\mrel})$ is acyclic (intuitively, messages cannot travel
  backwards in time), and the associated partial order is denoted
  ${\le} = ({\prel} \cup {\mrel})^*$ with strict part ${<}=({\prel} \cup {\mrel})^+$,

  \item each event is part of at most one message edge, and
\item for all $(p,q) \in \Ch$ and $(e,f), (e',f') \in {\mrel} \cap (E_p \times E_q)$, we have $e \prel^* e'$ iff $f \prel^* f'$ (which guarantees a FIFO behavior).
\eofex
\end{itemize}
\end{definition}

An event that does not belong to a message edge is called \emph{internal}. We say that two events $e,f \in E$ are \emph{parallel}, written $e \conc f$, if neither $e \le f$ nor $f \le e$. The set of all MSCs is denoted $\MSCs{\Procs}{\Sigma}$.

\begin{example}
An example MSC over $\Procs = \{p,q,r\}$ and $\Sigma = \{a,b,c\}$ is depicted in Figure~\ref{fig:partition}. That is, $w_p = aacaaaaa$, $w_r = aaaaaaaaaa$, and $w_q = abbaacaaa$ (note that $q$ is the bottom process). Consider the events $f=(p,4)$, $e = (p,5)$, and $g=(q,2)$. We have $f \prel e$ and $g \mrel e$. Moreover, $(p,3) \conc (q,6)$ (i.e., the two $c$-labeled events are parallel), while $(p,3) \le (q,8)$.
\end{example}

\begin{remark}
An MSC $M=((w_p)_{p \in \Procs},\mrel)$ is uniquely determined by $E$, $\prel$, $\mrel$, and the mapping $\labloc: E \to (\Procs \times \Sigma)$ defined by $\lambda((p,i)) = (p,a)$ where $a$ is the $i$-th letter of $w_p$.
Therefore, we will henceforth refer to $M$ as the tuple $M = (E,\prel,\mrel,\labloc)$.
\end{remark}

\begin{figure}
  \centering
  \begin{tikzpicture}[semithick,>=stealth,xscale=1.5,yscale=-1.2]
  \filldraw[color=blue!20,draw=blue] plot[smooth] coordinates
              { (-0.5,-0.3) (1.5,-0.3) (1.5,0.5) (2.1,0.5) (2.1,1.5) (1.2,1.5) (1.2,2.3) (0.8,2.3) (0.8,1.5) (0.2,1.5) (0.2,2.3) (-0.5,2.3) };

  \fill[red!30,draw=red] plot[smooth] coordinates
              { (1.4,2.35) (1.4,1.7) (2.23,1.6) (2.258,0.5) (2.7,0.5) (2.7,1.5) (3.65,1.5) (3.65,2.35) };

  \fill[mygreen!30,,draw=mygreen] plot[smooth] coordinates
              { (5.0,2.3) (3.8,2.3) (3.8,1.4) (2.85,1.4) (2.85,0.5) (3.3,0.5) (3.4,-0.3) (5.0,-0.3) };

  \fill[mygreen!30,rounded corners = 3pt] (2.7,-0.3) rectangle (2.3,0.3);
  \fill[red!30,rounded corners = 3pt] (2.12,-0.3) rectangle (1.7,0.4);
  \fill[red!30,rounded corners = 3pt] (3.2,-0.3) rectangle (2.8,0.3);
  \fill[red!30,rounded corners = 3pt] (0.7,1.6) rectangle (0.3,2.35);
  
    \draw (-0.5,0) -- (5,0);
    \draw (-0.5,1) -- (5,1);
    \draw (-0.5,2) -- (5,2);
    \node at (-0.7,0) {$p$};
    \node at (-0.7,1) {$r$};
    \node at (-0.7,2) {$q$};
    \node at (1,0.01) {\textbullet};
    \node at (3,2.01) {\textbullet};
    \node at (4.5,2.01) {\textbullet};
    
    \node at (0,-0.13) {\scalebox{0.8}{$a$}};
    \node at (0.5,-0.13) {\scalebox{0.8}{$a$}};
    \node at (1,-0.13) {\scalebox{0.8}{$c$}};
    \node at (2,-0.13) {\scalebox{0.8}{$a$}};
    \node at (2.5,-0.13) {\scalebox{0.8}{$a$}};
    \node at (3,-0.13) {\scalebox{0.8}{$a$}};
    \node at (3.5,-0.13) {\scalebox{0.8}{$a$}};
    \node at (4.5,-0.13) {\scalebox{0.8}{$a$}};

    \node at (-0.3,0.87) {\scalebox{0.8}{$a$}};
    \node at (0,1.13) {\scalebox{0.8}{$a$}};
    \node at (0.5,1.13) {\scalebox{0.8}{$a$}};
    \node at (1,0.87) {\scalebox{0.8}{$a$}};
    \node at (1.8,0.87) {\scalebox{0.8}{$a$}};
    \node at (2,1.13) {\scalebox{0.8}{$a$}};
    \node at (2.5,0.87) {\scalebox{0.8}{$a$}};
    \node at (3,1.13) {\scalebox{0.8}{$a$}};
    \node at (4,0.87) {\scalebox{0.8}{$a$}};
    \node at (4.5,1.13) {\scalebox{0.8}{$a$}};

    \node at (-0.3,2.15) {\scalebox{0.8}{$a$}};
    \node at (0.5,2.14) {\scalebox{0.8}{$b$}};
    \node at (1,2.14) {\scalebox{0.8}{$b$}};
    \node at (1.7,2.15) {\scalebox{0.8}{$a$}};
    \node at (2.5,2.15) {\scalebox{0.8}{$a$}};
    \node at (3,2.16) {\scalebox{0.8}{$c$}};
    \node at (3.5,2.15) {\scalebox{0.8}{$a$}};
    \node at (4,2.15) {\scalebox{0.8}{$a$}};
    \node at (4.5,2.15) {\scalebox{0.8}{$a$}};

    \draw[->] (0,0) -- (0,1);
    \draw[->] (0.5,0) -- (0.5,1);
    \draw[->] (1,2) -- (1,1);
    \draw[->] (-0.3,2) -- (-0.3,1);
    \draw[->] (0.5,2) -- (2.5,0);
    \draw[->] (2,1) -- (2,0);
    \draw[->] (1.7,1) -- (1.7,2);
    \draw[->] (2.5,1) -- (2.5,2);
    \draw[->] (3,0) -- (3,1);
    \draw[->] (3.5,2) -- (3.5,0);
    \draw[->] (4,1) -- (4,2);
    \draw[->] (4.5,1) -- (4.5,0);

    \node at (2.55,0.18) {\scalebox{0.9}{$e$}};
    \node[black!50!green] at (2.5,-0.49) {\scalebox{0.9}{$\apprelp{e}{=}$}};
    \node at (1.85,0.2) {\scalebox{0.9}{$f$}};
    \node at (0.45,1.8) {\scalebox{0.9}{$g$}};
    \node[blue] at (0.5,-0.55) {\scalebox{0.8}{$\apprelp{e}{\strictless^{-1}}$}};
    \node[black!50!green] at (4.4,2.55) {\scalebox{0.8}{$\apprelp{e}{\strictless}$}};
    \node[red] at (1.6,-0.5) {\scalebox{0.8}{$\apprelp{e}{\prel^{-1}}$}};
    \node[red] at (3.3,-0.5) {\scalebox{0.8}{$\apprelp{e}{\prel}$}};
    \node[red] at (0.5,2.55) {\scalebox{0.8}{$\apprelp{e}{\mrel^{-1}}$}};
    \node[red] at (2.5,2.55) {\scalebox{0.8}{$\apprelp{e}{\conc}$}};
  \end{tikzpicture}
  \caption{An MSC; the partition determined by an event $e$\label{fig:partition}}
\end{figure}
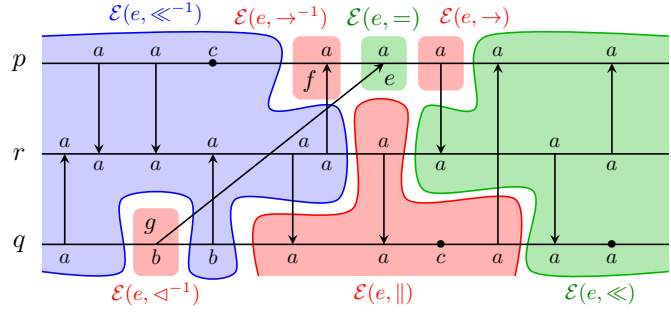

Let $M = (E,\prel,\mrel,\labloc)$ be an MSC over $\Procs$ and $\Sigma$.  A
\emph{run} of the \CFM $\A$ on $M$ is given by a mapping $\rho$ that associates with
every event $e \in E_p$ ($p\in\Procs$) the transition $\rho(e) \in\Delta_p$ that
is executed at $e$.  We require that
\begin{enumerate}
\item for every $e \in E$ with $\labloc(e) = (p,a)$, we have $\tlabel(\rho(e)) = a$,
\item for every process $p \in \Procs$ such that $E_p \neq \emptyset$, we have $\source(\rho((p,1))) = \init_p$,
\item for every process edge $(e,f) \in {\prel}$, we have $\target(\rho(e)) = \source(\rho(f))$,
\item for every internal event $e \in E$, $\rho(e)$ is an internal transition, and

\item for every message edge $(e,f) \in {\mrel}$ with $e \in E_p$ and $f \in
E_q$, $\rho(e)\in\Delta_p$ is a send transition and $\rho(f)\in\Delta_q$ is a
receive transition such that $\tmsg(\rho(e)) = \tmsg(\rho(f))$,
$\receiver(\rho(e)) = q$, and $\sender(\rho(f)) = p$.
\end{enumerate}
Note that, when $|\Procs| = 1$, Condition 5.\ becomes meaningless and Conditions 1.--4.\ emulate the behavior of a finite automaton.

It remains to define when $\rho$ is accepting. To this aim, we collect the final states of each process $p$. If $E_p \neq \emptyset$, then let $s_p$ be the target state of $\rho((p,|w_p|))$, i.e., of the last transition taken by $p$. Otherwise, let $s_p = \init_p$. Now, we say that $\rho$ is \emph{accepting} if $(s_p)_{p \in \Procs} \in \Acc$.

Finally, the \emph{language} of $\A$ is defined as $L(\A) = \{M \in \MSCs{\Procs}{\Sigma} \mid$ there is an accepting run of $\A$ on $M\}$.

\subsection{MSO and Two-Variable Logic}

While \CFMs serve as an operational model of concurrent systems, MSO logic can be considered as a high-level specification language. It uses first-order variables $x,y,\ldots$ to quantify over events, and second-order variables $X,Y,\ldots$ to represent sets of events.
The logic $\MSO$ is defined by the following grammar (recall that we have fixed $\Procs$ and $\Sigma$):
\[
  \varphi ::= p(x) \mid a(x) \mid x \in X \mid x = y
  \mid x \prel y \mid x \mrel y \mid x \le y
  \mid \varphi \lor \varphi \mid \lnot \varphi \mid \exists x.\varphi \mid \exists X.\varphi
\]
where $x$ and $y$ are first-order variables, $X$ is a second-order variable, $a
\in \Sigma$, and $p \in \Procs$.  For convenience, we allow usual
abbreviations such as conjunction $\varphi \wedge \psi$, universal
quantification $\forall x.\varphi$, implication $\varphi \Longrightarrow \psi$, etc.
The atomic formulas $p(x)$ and $a(x)$ are interpreted as ``$x$ is located on
process $p$'' and, respectively, ``the label of event $x$ is $a$''.  The binary
predicates are self-explanatory, and the boolean connectives and quantification
are interpreted as usual. The size $|\varphi|$ of a formula $\varphi \in \MSO$ is the length of $\varphi$ seen as a string.

A variable that occurs free in a formula requires an interpretation in terms of an event/a set of events from the given MSC. We will write, for example, $M,x \mapsto e,y \mapsto f \models \varphi$ if $M$ satisfies $\varphi$ provided $x$ is interpreted as $e$ and $y$ as $f$. If $\varphi$ is a \emph{sentence} (i.e., does not contain any free variable), then we write $M \models \varphi$ to denote that $M$ satisfies $\varphi$. With a sentence $\varphi$, we associate the MSC language $L(\varphi) = \{M \in \MSCs{\Procs}{\Sigma} \mid M \models \varphi\}$.

The set $\FO$ of \emph{first-order formulas} is the fragment of $\MSO$ that does not make use of second-order quantification $\exists X$.
The two-variable fragment of $\FO$, denoted by $\FOt$, allows only for two
first-order variables, $x$ and $y$ (which, however, can be quantified and reused arbitrarily
often).  Moreover, formulas from $\EMSO$, the \emph{existential fragment} of
$\MSO$, are of the form $\exists X_1 \ldots \exists X_n.\varphi$ where $\varphi
\in \FO$.  Accordingly, $\EMSOt$ is the set of $\EMSO$ formulas whose
first-order kernel is in $\FOt$.

The expressive power of all these fragments heavily depends on the set of binary predicates among $\{\prel,\mrel,\le\}$ that are actually allowed.
For a logic $\mathcal{C} \in \{\MSO,\EMSO,\EMSOt,\FO,\FOt\}$, and a set $\signature \subseteq \{\prel,\mrel,\le\}$, let $\mathcal{C}[\signature]$ be the logic $\mathcal{C}$ restricted to the binary predicates from $\signature$ (however, we always allow for equality, i.e., formulas of the form $x = y$). In particular, $\MSO = \MSO[\prel,\mrel,\le]$. As the transitive closure of a binary relation is definable in terms of second-order quantification, $\MSO[\prel,\mrel,\le]$ and $\MSO[\prel,\mrel]$ have the same expressive power (over MSCs). On the other hand, $\MSO[\le]$ is strictly less expressive \cite{BolligJournal}.

\begin{example}
Suppose $\Procs = \{p,q,r\}$ and $\Sigma = \{a,b,c\}$. The (mutual exclusion) formula $\neg \exists x. \exists y. (c(x) \wedge c(y) \wedge x \conc y)$, where $x \conc y$ is defined as $\neg (x \le y) \wedge \neg (y \le x)$, is in $\FOt[\le]$. It is not satisfied by the MSC from Figure~\ref{fig:partition}, as the two $c$-labeled internal events are parallel.
\end{example}

Let us turn to the relative expressive power of \CFMs and logic. We say that \CFMs and a logic $\mathcal{C}$ are \emph{expressively equivalent} if,
\begin{itemize}
\item for every \CFM $\A$, there exists a sentence $\varphi \in \mathcal{C}$ such that $L(\A) = L(\varphi)$, and
\item for every sentence $\varphi \in \mathcal{C}$, there exists a \CFM $\A$ such that $L(\A) = L(\varphi)$.
\end{itemize}

Now, the B{\"u}chi-Elgot-Trakhtenbrot theorem can be stated as follows:

\begin{theorem}[\!\!\cite{Buechi:60,Elgot1961,Trakhtenbrot62}]
If $|\Procs| = 1$, then \CFMs (i.e., finite automata) and $\MSO$ are expressively equivalent.
\end{theorem}

Unfortunately, when several processes are involved, \MSO is too expressive to be captured by \CFMs, unless one restricts the logic:

\begin{theorem}[\!\!\cite{BolligJournal}]
\CFMs and $\EMSO[\prel,\mrel]$ are expressively equivalent.
\end{theorem}

The logic $\EMSO[\prel,\mrel]$ is not very convenient as a specification language, as it does not allows us to talk, explicitly, about the order of an MSC.
It should be noted that \CFMs and \MSO are expressively equivalent if one restricts to MSCs that are \emph{channel-bounded} \cite{HenriksenJournal,GKM06,Kuske01}. Our main result allows one to include $\le$ in the unbounded case, too, though we have to restrict to two first-order variables:

\begin{theorem}\label{thm:equ}
\CFMs and $\EMSOt[\prel,\mrel,\le]$ are expressively equivalent.
\end{theorem}

Both directions are effective. Translating a \CFM into an \EMSOt formula is standard: Second-order variables represent an assignment of transitions to events. The first-order kernel then checks whether this guess is consistent with the definition of an accepting run.



\section{From Two-Variable Logic To \CFMs}\label{sec:result}

The rest of this paper is devoted to the translation of $\EMSOt[\prel,\mrel,\le]$ formulas into \CFMs.

\begin{theorem}
  \label{thm:main}
  For all sentences $\varphi \in \EMSOt[\prel,\mrel,\le]$, we can effectively construct a \CFM $\A_\varphi$ with $2^{2^{\mathcal{O}(|\varphi| + |P| \log{|P|)}}}$ states (per process) such that $L(\A_\varphi) = L(\varphi)$.
\end{theorem}

The \CFM $\A_\varphi$ is inherently \emph{nondeterministic} (for the definition of a deterministic \CFM, cf.\ \cite{GKM07}). Already for $\FOt$, this is unavoidable: \CFMs are in general not determinizable, as witnessed by an \FOt-definable language in \cite[Proposition~5.1]{GKM07}.
Note that the number of states of $\A_\varphi$ is, in fact, independent of the number of letters from $\Sigma$ that do \emph{not} occur in the formula. This is why Theorem~\ref{thm:main} mentions only $|\varphi|$ rather than $|\Sigma|$. Actually, the doubly exponential size of $\A_\varphi$ is necessary, even for $\FOt[\prel]$
or $\FOt[\le]$ sentences and a small number of processes. The following can be shown using known techniques \cite{GroheS03,Weis2011} (see Appendix~\ref{sec:lb}):
\begin{lemma}
  (i) Assume $|P| = 1$ and $|\Sigma| = 2$. For all $n \in \mathbb{N}$, there is a sentence $\varphi \in \FOt[\prel]$ of size $\mathcal{O}(n^2)$ such that no \CFM with less than
  $2^{2^{n}}$ states recognizes $L(\varphi)$.\\
  (ii) Assume $|P| = 2$ and $|\Sigma| = n$ with $n \ge 2$. There is a sentence $\varphi \in \FOt[\le]$ of size $\mathcal{O}(n)$ such that no \CFM with less than
  $2^{2^{n-1}}$ states on every process recognizes~$L(\varphi)$.
\end{lemma}

Now, we turn to the upper bound, i.e., the proof of Theorem~\ref{thm:main}. In a first step, we translate the given formula into Scott normal form:

\begin{lemma}[Scott Normal Form]
  Every formula from $\EMSOt[\prel,\mrel,\le]$ is effectively equivalent to a linear-size formula of the form
  $
    \exists X_1 \ldots \exists X_m.\psi$ where $\psi = \forall x. \forall y.
    \varphi \land \bigwedge_{i=1}^{\nform} \forall x. \exists y. \varphi_i \in \FOt[\prel,\mrel,\le]$ with
    $\varphi,\varphi_1,\ldots,\varphi_\nform$ quantifier-free.
\end{lemma}

As \CFMs are closed under projection, it remains to deal with the first-order part $\psi$.
Note that $\psi$ contains free occurrences of second-order variables $X_1,\ldots,X_\nvar$.
To account for an interpretation of these variables, we extend the alphabet $\Sigma$ towards the alphabet $\extSigma = \Sigma \times \{0,1\}^\nvar$ of exponential size. When an event $e$ is labeled with $(a,b_1,\ldots,b_\nvar) \in \extSigma$, we consider that $e \in X_i$ iff $b_i = 1$.
As \CFMs are closed under intersection, too, the proof of Theorem~\ref{thm:main} comes down to the translation of the formulas $\forall x. \forall y. \varphi$ and $\forall x. \exists y. \varphi_i$.

Notice that, given an MSC $M$ and events $e$ and $f$ in $M$,
whether $M, x \mapsto e, y \mapsto f \models \varphi$ holds or not
only depends on the labels of $e$ and $f$, and their relative position.
This is formalized below in terms of \emph{types}.

\newcommand{\formula}{\eta}
\newcommand{\rewr}[4]{\llbracket #1 \rrbracket_{#2,#3}^{#4}}
\newcommand{\true}{\mathit{true}}
\newcommand{\false}{\mathit{false}}

\paragraph{Types.}

Let $M =(E,\prel,\mrel,\lambda) \in \MSCs{\Procs}{\extSigma}$ be an MSC. Towards the definition of the type of an event, we define another binary relation ${\strictless} = {<} \setminus ({\prel} \cup {\mrel})$.
Let $\Rel$ be the set of \emph{relation symbols} $\{=,\prel,\mrel,\,\conc\,,\prel^{-1},\mrel^{-1},\strictless,\strictless^{-1}\}$. Given $e \in E$ and ${\bowtie} \in \Rel$, we let $\apprel{M}{e}{\bowtie} = \{f \in E \mid e \bowtie f\}$.
In particular, $\apprel{M}{e}{\strictless^{-1}} = \{f \in E \mid f < e \mathrel{\wedge} \neg(f \prel e) \mathrel{\wedge} \neg(f \mrel e)\}$. When $M$ is clear from the context, we may just write $\apprelp{e}{\bowtie}$.
Notice that all these sets form a partition of $E$, i.e., $E = \biguplus_{\bowtie \in \Rel} \apprelp{e}{\bowtie}$ (some sets may be empty, though).
The $\bowtie$-\emph{type} and the \emph{type} of an event $e \in E$ are respectively defined by
\[\type_M^{\bowtie}(e) = \{\lambda(f) \mid f \in \apprel{M}{e}{\bowtie}\}
\quad\text{and}\quad
\type_M(e) = \bigl(\type_M^{\bowtie}(e)\bigr)_{\bowtie \in \Rel}\,.\]
By $\Types = \prod_{{\bowtie} \in \Rel} 2^{\Procs \times \extSigma}$,
we denote the (finite) set of possible types. Thus, we deal with functions $\type_M^{\bowtie}: E \to 2^{\Procs \times \extSigma}$ and $\type_M: E \to \Types$.

\begin{example}
Consider Figure~\ref{fig:partition} and the distinguished event $e$.  Suppose
$a,b,c \in \extSigma$.  The sets $\apprelp{e}{\bowtie}$, which form a partition
of the set of events, are indicated by the colored areas.  Note that, since $e$
is a receive event, $\apprelp{e}{\mrel} = \emptyset$.  Moreover,
$\type_M^{\prel}(e) = \type_M^{=}(e) = \type_M^{\prel^{-1}}(e) = \{(p,a)\}$ and
$\type_M^{\strictless^{-1}}(e) = \{(p,a),(p,c),(r,a),(q,a),(q,b)\}$.
\end{example}

In fact, it is enough to know the type of every event to (effectively) evaluate
$\psi$.  To formalize this, let $\formula \in
\{\varphi,\varphi_1,\ldots,\varphi_\nform\}$.  Recall that $\formula$ has free
first-order variables $x$ and $y$.  Assume that we are given $M \in
\MSCs{\Procs}{\extSigma}$ and two events $e$ and $f$ that are labeled with
$(p,\sigma),(p',\sigma') \in \Procs \times \extSigma$, respectively, where
$\sigma = (a,b_1,\ldots,b_\nvar)$ and $\sigma' = (a',b_1',\ldots,b_\nvar')$.
Let ${\bowtie} \in \Rel$ be the unique relation such that $e \bowtie f$.
To decide whether $M,x \mapsto e,y \mapsto f \models \formula$, we rewrite
$\formula$ into a propositional formula
$\rewr{\formula}{(p,\sigma)}{(p',\sigma')}{\bowtie}$ that can be evaluated to
$\true$ or $\false$: Replace the formulas $p(x)$, $a(x)$, $p'(y)$, $a'(y)$, $x \in X_i$ with $b_i = 1$, and $y \in X_i$ with $b_i' = 1$ by $\true$.  All other unary
predicates become $\false$ (we consider $z \in X_i$ to be unary).  Formulas $z
\sim z'$ with $z,z' \in \{x,y\}$ and ${\sim} \in \{=,\prel,\mrel,\le\}$ can be
evaluated to $\true$ or $\false$ based on the assumption that $x \bowtie y$.
By an easy induction, we obtain: 

\begin{lemma}\label{lem:eval}
For all $\formula \in \{\varphi,\varphi_1,\ldots,\varphi_\nform\}$, $M=(E,\prel,\mrel,\lambda) \in \MSCs{\Procs}{\extSigma}$, and $e \in E$:\vspace{0.5ex} 
\begin{itemize}\itemsep=0.7ex
  \item $M, x \mapsto e \models \exists y.\formula$ iff 
  $\rewr{\formula}{\lambda(e)}{(p',\sigma')}{\bowtie}$ is true
  for some ${\bowtie} \in \Rel$ and $(p',\sigma') \in \type_M^{\bowtie}(e)$.

  \item $M, x \mapsto e \models \forall y.\formula$ iff 
  $\rewr{\formula}{\lambda(e)}{(p',\sigma')}{\bowtie}$ is true
  for all ${\bowtie} \in \Rel$ and $(p',\sigma') \in \type_M^{\bowtie}(e)$.
\end{itemize}
\end{lemma}

Therefore, in order to construct a \CFM for $\psi$, we start by constructing a \CFM $\Atype$ that
``labels'' each event with its type. 

In fact, our translation of a formula into a \CFM relies on several intermediate \CFMs running on \emph{extended} MSCs, whose events have additional labels from a finite alphabet $\Gamma$. It will be convenient to consider an extended MSC from $\MSCs{\Procs}{\extSigma \times \Gamma}$, in the obvious way, as a pair $(M,\addlab)$ where $M=(E,\prel,\mrel,\lambda) \in \MSCs{\Procs}{\extSigma}$ and $\addlab: E \to \Gamma$.

\begin{theorem}
  \label{thm:Atype}
  There is a \CFM $\Atype$ over $\Procs$ and $\extSigma \times \Types$
  with $2^{|\extSigma| \cdot 2^{\mathcal{O}(|P|\log{|P|})}}$ states
  such that $L(\Atype) = \{(M,\type_M) \mid M \in \MSCs{\Procs}{\extSigma}\}$.
\end{theorem}

According to Lemma~\ref{lem:eval},
the \CFM for $\forall x.\exists y.  \formula$ (respectively, $\forall x.\forall
y.  \formula$) is obtained from $\Atype$ by restricting the transition relation:
We keep a transition of process $p$ with label $(\sigma,(\ttype_{\bowtie})_{\bowtie
\in \Rel}) \in \extSigma \times \Types$ if
$\rewr{\formula}{(p,\sigma)}{(p',\sigma')}{\bowtie}$ is $\true$ for some
(respectively, for all) ${\bowtie} \in \Rel$ and $(p',\sigma') \in \ttype_{\bowtie}$.
Moreover, the new transition label will just be $\sigma$ (the type is projected away).

\medskip

We obtain $\Atype$ as the product of \CFMs $\A^{\bowtie}$ over $\Procs$ and $\extSigma \times 2^{\Procs \times \extSigma}$ such that $L(\A^{\bowtie}) = \{(M,\type_M^{\bowtie}) \mid M \in \MSCs{\Procs}{\extSigma}\}$.
Thus, it only remains to construct $\A^{\bowtie}$, for all ${\bowtie} \in \Rel$.
The cases ${\bowtie} \in \{=,\prel,\mrel,\prel^{-1},\mrel^{-1}\}$ are straightforward and can be found in Appendix~\ref{app:Abowtie}.
Below, we show how to construct $\A^{\strictless^{-1}}$.
We then obtain $\A^{\strictless}$ by symmetry.
The case $\A^{\parallel}$ is more difficult and will be treated in the
next section.

\begin{lemma}\label{lem:cfmstrictless}
  There is a \CFM $\A^{\strictless^{-1}}$ over $\Procs$ and $\extSigma \times 2^{\Procs \times \extSigma}$ with $2^{\mathcal{O}(|P \times \extSigma|)}$ states
  such that $L(\A^{\strictless^{-1}}) = \{(M,\type_M^{\strictless^{-1}}) \mid M \in \MSCs{\Procs}{\extSigma}\}$.
\end{lemma}

\begin{proof}
We sketch the idea, a detailed exposition can be found in Appendix~\ref{app:Astrictless}.
Consider Figure~\ref{fig:partition} and suppose $a,b,c \in \extSigma$.  At the
time of reading event $e$, the \CFM $\A^{\strictless^{-1}}$ should deduce
$\type_M^{\strictless^{-1}}(e)=\{(p,a),(p,c),(r,a),(q,a),(q,b)\} =: \ttype$.  To do
so, it collects all labelings from $\Procs \times \extSigma$ that it has seen in
the past (which is $\ttype$ when reading $e$).  Naively, one would then just remove
the labels $(p,a)$ and $(q,b)$ of the predecessors $f$ and $g$ of $e$.  However,
this leads to the wrong result, since both $(p,a)$ and $(q,b)$ are contained in
$\type_M^{\strictless^{-1}}(e)$.  In particular, there is another
$(q,b)$-labeled event $g' \in \apprelp{e}{\strictless^{-1}}$.  The solution is
to count the number of occurrences of each label up to $2$.  When reading $e$,
the \CFM will have seen $(p,a)$ and $(q,b)$ at least twice so that it can safely
conclude that both are contained in $\type_M^{\strictless^{-1}}(e)$.
\end{proof}



\section{Labels of Parallel Events}\label{sec:parallel}

In this section, we construct the \CFM $\A^{\conc}$ such that
$L(\A^{\conc}) = \{(M,\type_M^{\conc}) \mid M \in \MSCs{\Procs}{\extSigma}\}$.
This completes the proof of Theorem~\ref{thm:Atype} and, thus, of
Theorem~\ref{thm:main}. We obtain $\A^{\conc}$ as the product of
several \CFMs $\A_{p,q,a}$:

\begin{lemma}\label{lem:parallel}
  For all $p,q \in P$ with $p\neq q$ and $a \in \extSigma$, there is a \CFM $\A_{p,q,a}$ over $\Procs$ and $\extSigma \times \{0,1\}$ with $2^{2^{\mathcal{O}(|P|\log{|P|})}}$ states such that
\begin{multline*}
    L(\A_{p,q,a}) = \bigl\{ (M = (E,\prel,\mrel,\lambda),\gamma) \in \MSCs{\Procs}{\extSigma \times \{0,1\}} ~\mid \\
    \forall e \in E_p: \bigl(\gamma(e) = 1 \iff (q,a) \in \type_M^{\conc}(e)\bigr)\bigr\}\,.
\end{multline*}
\end{lemma}

The rest of this section is devoted to the proof of Lemma~\ref{lem:parallel}.

\medskip

Fix $p,q \in P$ ($p \neq q$) and $a \in \extSigma$. We construct $\A_{p,q,a}$ as the product (intersection) of two \CFMs
$\Azeros$ and $\Aones$ over $\Procs$ and $\extSigma \times \{0,1\}$, recognizing respectively the languages
\[\begin{array}{l}
\Lzeros = \bigl\{ (M = (E,\prel,\mrel,\lambda),\gamma) ~\mid~ \forall e \in E_p: \bigl(\gamma(e) = 0 \implies (q,a) \not\in \type_M^{\conc}(e)\bigr)\bigr\} \text{~~and}\\[1.5ex]
\Lones = \bigl\{ (M = (E,\prel,\mrel,\lambda),\gamma) ~\mid~ \forall e \in E_p: \bigl(\gamma(e) = 1 \implies (q,a) \in \type_M^{\conc}(e)\bigr)\bigr\}\,.\end{array}\]

\subsection{Construction of $\boldsymbol{\Azeros}$}

We first turn to the easier case of building $\Azeros$. Essentially, $\Azeros$ has to guess a path in an MSC that covers all $0$-events on $p$ as well as all $(q,a)$-events on $q$.

\begin{lemma}
  \label{lem:zeros}
  Let $(M,\gamma) \in \MSCs{\Procs}{\extSigma \times \{0,1\}}$ be an MSC with  $M = (E,\prel,\mrel,\lambda) \in \MSCs{\Procs}{\extSigma}$ and $\gamma : E \to \{0,1\}$.
  The following are equivalent:
  \begin{enumerate}
  \item $(M,\gamma) \in \Lzeros$.
  \item There is a path $\chemin$ in $M$
    (i.e., a path in the directed graph $(E, {\prel} \cup {\mrel})$)
    such that
    all events $e$ on process $p$ with $\gamma(e) = 0$
    and all events $f$ such that $\labloc(f) = (q,a)$ are on~$\chemin$.
  \end{enumerate}
\end{lemma}

\begin{proof}
  We first show 1.\ $\implies$ 2.
  Let $E'_p = \{ e \in E_p \mid \gamma(e) = 0\}$
  and $E'_q = \{ f \in E_q \mid \lnot (\exists e \in E'_p: e \conc f) \}$.
  By assumption, $E'_q$ contains all events 
  $f$ such that $\labloc(f) = (q,a)$.  Let $E' = E'_p \cup E'_q$.  For all
  events $e, f \in E'$, either $e$ and $f$ are on the same process, or one is in
  $E'_p$ and the other in $E'_q$; in both cases, we have either $e \le f$ or $f
  \le e$.  So events in $E'$ are totally ordered wrt.\
  ${\leq}=({\prel}\cup{\mrel})^*$.  Hence there exists a path in $M$ connecting
  all events of~$E'$.
  
  Now assume Condition 2.\ is satisfied. Let $e$ be some event
  on process $p$ such that $\gamma(e) = 0$. Let $f$ be any event such that
  $\labloc(f) = (q,a)$.  By definition,
  both $e$ and $f$ are on path $\chemin$, so either $e \le f$ or $f \le e$.
  Thus, $e$ is not parallel to $f$. We deduce $(q,a) \not\in \type_M^{\conc}(e)$ and, therefore, $(M,\gamma) \in \Lzeros$.
\end{proof}

\begin{lemma}
  \label{lem:Azeros}
  There is a \CFM $\Azeros$ with a constant number of states such that $L(\Azeros) = \Lzeros$.
\end{lemma}

\begin{proof}
  The \CFM $\Azeros$ will try to guess a path $\chemin$ as in
  Lemma~\ref{lem:zeros}.
  This path is represented by a token moved along the MSC.
  Initially, 
  exactly one process has the token. At each event, the automaton
  may chose to pass along the token to the next event of the current process,
  or (if the event is a write) to send the token to another process.
  Formally, (non)-possession of the token is represented by two states,
  $s_{\mathsf{token}}$ and $s_{\overline {\mathsf{token}}}$, and movements of the
  token from one process to another by messages. 
  All global states are accepting.

  Process~$p$ may read an event labeled $0$ only if it has the token,
  and process~$q$ may read $a$'s only if it has the token,
  so that the path along which the token is moved contains all events $e$
  on process $p$ such that $\gamma(e) = 0$, and all events $f$ such that
  $\labloc(f) = (q,a)$.
  
  Clearly, $\Azeros$ has an accepting run on $M$ iff there exists a
  path in $M$ as described in Lemma~\ref{lem:zeros}.
\end{proof}

\subsection{Construction of $\boldsymbol{\Aones}$}

Let $M = (E,\prel,\mrel,\lambda)$ be an MSC. For $e \in E$ and $F \subseteq E$,
let $\Parallelp p e = \{f \in E_p \mid f \conc e\}$ and $\Parallelp p F = \{e \in E_p \mid e \conc f$ for some $f \in F\}$.
Moreover, given $e \in E$, define
$\rPastp p e  = \{f \in E_p \mid f < e\}$ and
$\rFuturep p e  = \{f \in E_p \mid e < f\}$.
An \emph{interval} in $M$ is
a (possibly empty) finite set of events $\{e_1, \ldots, e_k\}$ such that
$e_1 \prel \cdots \prel e_k$.
For all $e,f \in E_p$, we denote by $[e,f]$ the interval
$\{g \in E_p \mid e \le g \le f\}$.

\begin{remark}
  For all $p \in \Procs$ and $e \in E$, the sets $\rPastp p e$, $\Parallelp p e$, and
  $\rFuturep p e$ are intervals (possibly empty) of events on process $p$,
  such that $E_p = \rPastp p e \uplus \Parallelp p e \uplus \rFuturep p e$.
\end{remark}

The idea is that $\Aones$ will guess a set of intervals covering all
$1$-labeled events on process~$p$, and check that, for each interval $I$,
there exists an event $f$ 
such that $\labloc(f) = (q,a)$
and $I = \Parallelp p f$.

We first show that it will be sufficient for $\Aones$ to guess \emph{disjoint}
intervals (or more precisely, two sequences of disjoint intervals):
\begin{lemma}
  Let $M = (E,\prel,\mrel,\labloc) \in \MSCs{\Procs}{\extSigma}$
  and $F = \{f \in E \mid \labloc(f) = (q,a)\}$.
  There exist subsets $F_1,F_2 \subseteq F$
  such that the following hold:
  \begin{itemize}
  \item
    $\Parallelp p {F_1} \cup \Parallelp p {F_2} = \Parallelp p {F}$.
  \item For $i \in \{1,2\}$, the intervals in $\Parallelp p {F_i}$ are
    pairwise disjoint, and not adjacent: if $f,f' \in F_i$ and $f \neq f'$,
    then $\Parallelp p f \cup \Parallelp p {f'}$ is not an interval.
  \end{itemize}
\end{lemma}

\begin{proof}
  We first construct a set $F'\subseteq F$ by iteratively removing events from
  $F$, until there remains no event $f$ such that
  $\Parallelp p {f} \subseteq \Parallelp p {F' \setminus \{f\}}$.
  This ensures that, for each event $f\in F'$,
  there is at most one event $f'\in F'$ such that $f<f'$ and 
  $\Parallelp{p}{f}\cup\Parallelp{p}{f'}$ is an interval.
  Indeed, consider three events $f,f',f''\in E_q$ such that $f<f'<f''$ and 
  $\Parallelp{p}{f}\cup\Parallelp{p}{f''}$ is an interval. Then, 
  $\Parallelp{p}{f'}\subseteq\Parallelp{p}{f}\cup\Parallelp{p}{f''}$ and these 
  three events cannot all be in $F'$.
  
  Since, for each event $f\in F'$, there is at most one event $f'\in F'$ such
  that $f<f'$ and $\Parallelp{p}{f}\cup\Parallelp{p}{f'}$ is an interval, the
  set $F'$ can be divided into two sets $F_1$ and $F_2$ satisfying the
  requirements of the lemma.
\end{proof}

So, $\Aones$ will proceed as follows.
It will guess the sets $F_1$, $F_2$, $\Parallelp p {F_1}$ and
$\Parallelp p {F_2}$,
that is, label some events on process $q$ with ``$F_1$'' or ``$F_2$'',
and some events on process $p$ with ``$F_1$'' and/or ``$F_2$''.
This labeling must be such that on process $q$, only events initially labeled
$a$ may be labeled ``$F_1$'' or ``$F_2$'' (the sets guessed for $F_1$ and $F_2$
contain only events labeled~$a$), and that on process $p$,
all events initially labeled $1$ must be labeled either ``$F_1$'', ``$F_2$'',
or both (the sets guessed for $\Parallelp p {F_1}$ and $\Parallelp p {F_2}$ cover
all events labeled $1$ on process~$p$).
Then, $\Aones$ will check in parallel that both sets of marked events
(that is, either with ``$F_1$'', or with ``$F_2$'') satisfy the following
property: for every non-empty maximal interval $I$ of marked events on process $p$,
there exists a marked event $f$ on process $q$ such that $I = \Parallelp p f$.
Clearly, if $\Aones$ has an accepting run on $M$, then $M \in \Lones$.
Conversely, if $M \in \Lones$, then if $\Aones$ guesses correctly the sets
$F_1$, $F_2$, $\Parallelp p {F_1}$ and $\Parallelp p {F_2}$, it accepts.

\paragraph{The MSC language $\boldsymbol{\Lo}$.}

The different labelings $F_1$ and $F_2$ can be dealt with by two separate \CFMs so that we can restrict to a single labeling. More precisely, we will henceforth consider MSCs $(M,\gamma)$ with $\gamma: E \to \{0,1\}$ where the $1$-labeled events form a collection of maximal intervals on process $p$ and a set of events on process $q$. Now, the construction of $\Aones$ boils down to the construction of an automaton $\Ao$ recognizing the language $\Lo$: Let $\Lo$ be the set of MSCs $(M,\gamma)$ with $\gamma : E \to \{0,1\}$ such that
\begin{itemize}
\item for each non-empty maximal interval $I$ of $1$-labeled events on process $p$, there exists
a $1$-labeled event $f$ on process $q$ such that $\Parallelp p f = I$, and
\item conversely, for all $1$-labeled events $f$ on process $q$,
there exists a non-empty maximal interval $I$ of $1$-labeled events on process $p$
such that $\Parallelp p f = I$.
\end{itemize}
Note that we include the second condition only for technical reasons.

\medskip

We can decompose this problem one last time.
Let $\Paths$ (respectively, $\Paths_{p,q}$) be the set of process sequences
$\p = p_1 \ldots p_n$ (respectively, with $p_1 = p$ and $p_n = q$) such that
$n \ge 1$ and $p_i \neq p_j$ for $i \neq j$.
For all $\p=p_1 \ldots p_n \in \Paths$, we write $e \le_\p f$ if there exist events
$e = e_1, f_1, e_2, f_2, \ldots, e_n, f_n = f$ such that, for all $i$,
we have $e_i,f_i \in E_{p_i}$, $e_i \prel^* f_i$, and $f_i \mrel e_{i+1}$.
For all events $e \in E$ such that $\{ f \in E \mid f \le_\p e \}$
(respectively, $\{ f \in E \mid e \le_\p f \}$) is non-empty, we let
\[
  \last \p e  = \max \{ f \in E \mid f \le_\p e \} \quad\text{~and~}\quad
  \first \p e  = \min \{ f \in E \mid e \le_\p f \} \, .
\]
This is well-defined since all events in $\{ f \in E \mid f \le_\p e \}$
(respectively, $\{ f \in E \mid e \le_\p f \}$) are on the same process, hence are 
ordered.
Note that, if $\p = p$ consists of a single process, then, for all $e \in E_p$,
we have $\last \p e = e = \first \p e$.
Moreover, notice that ${\le} = \bigcup_{\p \in \Paths} \le_\p$.

\medskip

Let $\Lint$ be the set of MSCs $(M,\gamma)$
    where the mapping $\gamma : E \to \{0,1\}$ defines (non-empty maximal) intervals
    ${[e_1,e_1']}, \ldots, {[e_k,e_k']}$ of $1$-labeled events on process $p$ and
    a sequence of $1$-labeled events $f_1<\cdots<f_k$ on process $q$, such that,
    for all $1 \le i \le k$, we have $\rPastp p {e_i} \subseteq \rPastp p {f_i}$
    and $\rFuturep p {e'_i} \subseteq \rFuturep p {f_i}$. This is illustrated in Figure~\ref{fig:Aint}. Note that $\Lo \subseteq \Lint$. The converse inclusion does not hold in general, since the intervals in MSCs from $\Lint$ may be too large. However, we obtain $\Lo$ when we restrict $\Lint$ further to the intersection of the following two languages:
\begin{itemize}

\item $\Lleft$ is the set of all MSCs in $\Lint$ such that, for all
    $1 \le i \le k$ and $\p \in \Paths_{p,q}$, if $\last \p {f_i}$
    is defined, then $\last \p {f_i} \notin {[e_i,e'_i]}$.
    
\item $\Lright$ is the set of all MSCs in $\Lint$ such that, for all
    $1 \le i \le k$ and $\p \in \Paths_{q,p}$, if $\first \p {f_i}$
    is defined, then $\first \p {f_i} \notin {[e_i,e'_i]}$.
\end{itemize}

\begin{lemma}
  We have $\Lo = \Lleft \cap \Lright$.
\end{lemma}

\begin{proof}
  Let $(M,\gamma) \in \Lo$ and $i \in \{1,\ldots,k\}$. 
  By definition, $[e_i,e'_i] = \Parallelp p {f_i}$.
  Since $[e_i,e_i']$ is non-empty, we have $\rPastp p {e_i} = \rPastp p {f_i}$
  and $\rFuturep p {e_i'} = \rFuturep p {f_i}$.
  Hence, $(M,\gamma) \in \Lleft \cap \Lright$.

  Now, let $(M,\gamma) \in \Lleft \cap \Lright$.
  Since $(M,\gamma) \in \Lint$, we have $\Parallelp p {f_i} \subseteq [e_i,e_i']$ for all $i$.
  Assume that there is $e \in [e_i,e'_i]$ such that
  $e \notin \Parallelp p {f_i}$, for instance $e \le f_i$.
  Then, there exists $\p\in\Paths_{p,q}$ such that $e \le_\p f_i$, hence $e \le \last \p {f_i}$.
  As $(M,\gamma) \in \Lleft$, we get $e'_i < \last \p {f_i}$.
  And since $\rFuturep p {e_i'} \subseteq \rFuturep p {f_i}$, we have
  $f_i < \last \p {f_i}$, a contradiction.
\end{proof}

 \begin{figure}
 \centering
    \begin{tikzpicture}
      \node[blue,dot,label=above:{\textcolor{blue}{$e_1$}},
      label=below:{\textcolor{blue}{1}}] (e1) at (0,2) {};
      \node[blue,dot,label=above:{\textcolor{blue}{$e'_1$}},
      label=below:{\textcolor{blue}{1}}] (ep1) at (2,2) {};
      \node[dot,label=below:{0}] (g1) at (2.5,2) {};
      \node[dot,label=below:{0}] (h1) at (4.5,2) {};
      \node[red,dot,label=above:{\textcolor{red}{$e_2$}},
      label=below:{\textcolor{red}{1}}] (e2) at (5,2) {};
      \node[red,dot,label=above:{\textcolor{red}{$e'_2$}},
      label=below:{\textcolor{red}{1}}] (ep2) at (7,2) {};
      \node[dot,label=below:{0}] (g2) at (7.5,2) {};
      \node[dot,label=below:{0}] (h2) at (8.5,2) {};
      \node[mygreen,dot,label=above:{\textcolor{mygreen}{$e_3$}},
      label=below:{\textcolor{mygreen}{1}}] (e3) at (9,2) {};
      \node[mygreen,dot,label=above:{\textcolor{mygreen}{$e'_3$}},
      label=below:{\textcolor{mygreen}{1}}] (ep3) at (11,2) {};
      \node[dot,label=below:{0}] (g3) at (11.5,2) {};
      \coordinate (i1) at (0,0);
      \node[dot,label=below:{$f_1$}] (f1) at (1,0) {};
      \coordinate (j1) at (1.5,0);
      \coordinate (i2) at (5.5,0);
      \node[dot,label=below:{$f_2$}] (f2) at (6,0) {};
      \coordinate (j2) at (6.5,0);
      \coordinate (i3) at (9.5,0);
      \node[dot,label=below:{$f_3$}] (f3) at (10,0) {};
      \coordinate (j3) at (10.5,0);

      \path[very thick]
      (i1) edge[blue] (j1)
      (j1) edge[blue,decorate,decoration=snake,->] (g1)
      (g1) edge[blue] (h1)
      (h1) edge[red,decorate,decoration=snake,->] (i2)
      (i2) edge[red] (j2)
      (j2) edge[red,decorate,decoration=snake,->] (g2)
      (g2) edge[red] (h2)
      (h2) edge[mygreen,decorate,decoration=snake,->] (i3)
      (i3) edge[mygreen] (j3)
      (j3) edge[mygreen,decorate,decoration=snake,->] (g3)
      (g3) edge[mygreen] (12,2);

      \draw (j1) -- (i2) (j2) -- (i3) (j3) -- (12,0);
      \draw (e1) -- (g1) (h1) -- (g2) (h2) -- (g3);

      \node at (-0.5,2) {$p$};
      \node at (-0.5,0) {$q$};

      \node at (1,1.7) {\textcolor{blue}{$\cdots$}} ;
      \node at (3.5,1.7) {$\cdots$} ;
      \node at (6,1.7) {\textcolor{red}{$\cdots$}} ;
      \node at (8,1.7) {$\cdots$} ;
      \node at (10,1.7) {\textcolor{mygreen}{$\cdots$}} ;
    \end{tikzpicture}
    \caption{Constructions of $\Aint$ and $\Ao$.\label{fig:Aint}}
  \end{figure}

\paragraph{A \CFM for $\boldsymbol{\Lo}$.}

The last piece of the puzzle is a \CFM $\Ao$ such that $L(\Ao) = \Lo$.  It is
built as the product (intersection) of \CFMs $\Aint$, $\Aleft$, and $\Aright$.

\begin{lemma}
  \label{lem:Lint}
  There is a \CFM $\Aint$ with a constant number of states such that we have
  $L(\Aint) = \Lint$.
\end{lemma}

\begin{proof}
  Again, we implement a sort of token passing, which is illustrated in
  Figure~\ref{fig:Aint}.  The token starts on process $p$ iff the first
  $p$-event is labeled $0$; otherwise, it must start on $q$.  Similarly, the
  token ends on process $p$ iff the last $p$-event is labeled $0$; otherwise, it
  must end on $q$. 
  Process~$p$ reads $0$'s when it holds the token, and $1$'s when it does not.
  Moreover, after sending the token, process $p$ must read some 1-labeled
  events.
  When sent by $p$ (respectively $q$), the token must reach $q$ (respectively $p$) before 
  returning to $p$ (respectively $q$).
  Finally, process $q$ reads only $0$-labeled events when it does not hold the
  token.  Moreover, process $q$ checks that, within every maximal interval where
  it holds the token, there is exactly one $1$-labeled event.
  
  It is easy to check that $(M,\gamma) \in \Lint$ iff there exists a path along 
  which the token is passed and satisfying the above conditions.
\end{proof}

We now show that there exists a \CFM $\Aleft$ that accepts an MSC
$(M,\gamma) \in \Lint$ iff $(M,\gamma) \in \Lleft$.
The idea is that  $\Aleft$ guesses a coloring of the intervals
of marked events
such that checking $\last \p {f_i} \notin [e_i,e_i']$
can be replaced with checking that $\last \p {f_i}$ is not in an interval
with the same color as $[e_i,e_i']$.
We need to prove that such a coloring exists, 
and that the colors associated
with the $\last \p {f_i}$ can be computed by the \CFM.

\begin{lemma}
  \label{lem:coloring}
  Let $(M,\gamma) \in \Lleft$, let $I_1, \ldots, I_k$ be the sequence of maximal
  intervals of events labeled $1$ on process $p$, and $f_1<\cdots<f_k$
  the corresponding events labeled $1$ on process $q$.
  There exists a coloring
  $\chi : \{1,\ldots,k\} \to \{1, \ldots, |\Paths_{p,q}| + 1 \}$
  such that, for all $i,j \in \{1,\ldots,k\}$ and $\p \in \Paths_{p,q}$,
  $\chi(i) = \chi(j)$ implies $\last \p {f_j} \notin I_i$.
\end{lemma}

\begin{proof}
  We write $i \rightsquigarrow j$ when there exists $\p \in \Paths_{p,q}$
  such that $\last \p {f_j} \in I_i$.
  Notice that if $i \rightsquigarrow j$, then $f_i<f_j$ 
  (otherwise, we would have $\last\p{f_j}<f_j<f_i$, but $\last\p{f_j} \in I_i$).
  So we can define $\chi$ by successively choosing colors for
  $1, \ldots, k$: For all $j$, it suffices to choose a color
  $\chi(j)\in\{1,\ldots,|\Paths_{p,q}| + 1 \}$ distinct from the at most
  $|\Paths_{p,q}|$ colors of indices
  $i < j$ such that $i \rightsquigarrow j$.
\end{proof}

\begin{lemma}
  \label{lem:last-label}
  Let $\Theta$ be a finite set.
  There exists a (deterministic) \CFM with $|\Theta|^{\mathcal{O}(|P|!)}$ states recognizing the set of doubly extended
  MSCs $(M,\theta,\xi)$ such that, for all events $e$, $\xi(e)$ is the partial
  function from $\Paths$ to $\Theta$ such that $\xi(e)(\p) = \theta(\last \p
  e)$.
\end{lemma}

\begin{proof}
  The \CFM stores the label $\xi(e)$ of an event $e$ in its state, and includes
  it in the message if $e$ is a send event.  At an event $e$ on process $\nproc$, the
  CFM checks that $\xi(e)(\nproc) = \theta(e)$.  Moreover, the CFM checks that:
  \begin{itemize}
  \item If $e$ has no predecessor, then $\xi(e)(\p)$ is undefined for all $\p\neq \nproc$.

  \item If $e$ has one $\prel$-predecessor $f$ but no $\mrel$-predecessor,
    then $\xi(e)(\p) = \xi(f)(\p)$ for $\p \neq \nproc$.
  \item If $e$ has one $\mrel$-predecessor $g$ on process $r$, but no $\prel$-predecessor, then
    $\xi(e)(\p r \nproc) = \xi(g)(\p r)$,
    and $\xi(e)(\p)$ is undefined if $\p \neq \nproc$ and $\p$ does not end with
    $r \nproc$.
  \item If $e$ has one $\prel$-predecessor $f$ and one $\mrel$-predecessor $g$ on process $r$,
    then $\xi(e)(\p r \nproc) = \xi(g)(\p r)$, and $\xi(e)(\p) = \xi(f)(\p)$
    if $\p \neq \nproc$ and $\p$ does not end with $r \nproc$.
    \qedhere
  \end{itemize}
\end{proof}

\begin{lemma}
  \label{lem:Aleft}
  There is a \CFM $\Aleft$ with $2^{2^{\mathcal{O}(|P| \log {|P|})}}$ states
  such that we have $L(\Aleft) \cap \Lint = \Lleft$.
\end{lemma}

\begin{proof}
  Let $(M,\gamma) \in \Lint$ with $I_1,\ldots,I_k$ the non-empty maximal intervals of
  $1$-labeled events on process~$p$, and $f_1<\cdots<f_k$ the corresponding
  $1$-labeled events on process~$q$.

  We can slightly modify $\Aint$ so that on input $(M,\gamma)$, it guesses a
  coloring $\chi : \{1,\ldots,k\} \to \{1, \ldots, |\Paths_{p,q}| + 1 \}$,
  and labels each event in $I_i$ with $\chi(i)$.
  The color of the upcoming interval $I_i$
  is passed along with the token, so that at each $f_i$,
  the \CFM has access to the color $\chi(i)$ (see Figure~\ref{fig:Aint}).

  We can then compose that automaton with the \CFM from
  Lemma~\ref{lem:last-label}, to compute, at each $f_i$ and for all
  $\p \in \Paths_{p,q}$, the color associated with $\last \p {f_i}$.
  The \CFM $\Aleft$ then checks that for all $i$ and $\p$, either $\last \p {f_i}$
  is undefined, or $\gamma(\last \p {f_i}) = 0$,
  or the color associated with $\last \p {f_i}$ is different from $\chi(i)$.
  
  Suppose $(M,\gamma) \in L(\Aleft) \cap \Lint$. Then, for all $i$ and $\p$,
  $\last \p {f_i}$ cannot be in an interval colored $\chi(i)$.
  In particular, this implies $\last \p {f_i} \notin {I_i}$.
  Conversely, suppose $(M,\gamma) \in \Lleft$. Then, by Lemma~\ref{lem:coloring}, there exists
  a run in which the coloring guessed along the token passing is such that $\Aleft$ accepts.
\end{proof}

Finally, we obtain $\Ao$ as the product (intersection) of $\Aint$, $\Aleft$, and the mirror
$\Aright$ of $\Aleft$, which recognizes $\Lright$.  In fact, it is easy to see
that \CFMs are closed under mirror languages, in which both the process and the
edge relations are inverted.

\begin{lemma}
  There is a \CFM $\Ao$ with $2^{2^{\mathcal{O}(|P| \log{|P|})}}$ states such that $L(\Ao) = \Lo$.
\end{lemma}



\section{Conclusion}\label{sec:conclusion}

We showed that every $\EMSOt$ formula over MSCs can be effectively translated into an equivalent \CFM of doubly exponential size, which is optimal. At the heart of our construction is a \CFM $\Atype$ of own interest, which ``outputs'' the type of each event of an MSC. In particular, $\Atype$ can be applied to other logics such as propositional dynamic logic (PDL), which combines modal operators and regular expressions \cite{FisL79}. It has been shown in \cite{BKM-lmcs10} that every PDL formula can be translated into an equivalent \CFM. We can extend this result by adding a modality $\langle\conc\rangle$ to PDL, which ``jumps'' to some parallel event. For example, the formula $\neg\mathsf{E}(\mathit{CS} \wedge \langle\conc\rangle\mathit{CS})$ says that no two parallel events access a critical section. Note that \cite{BKM-lmcs10} considers infinite MSCs. However, it is easy to see that all our constructions can be extended to infinite MSCs.

A major open problem is whether every sentence from $\FO[\prel,\mrel,\le]$, with arbitrarily many variables, is
equivalent to some \CFM. To the best of our knowledge, the question is even open
for the logic $\FO[\le]$.  Generally, it would be worthwhile to identify large
classes of acyclic graphs of bounded degree such that all $\FO$- or
$\FOt$-definable languages (including the transitive closure of the edge
relation) are ``recognizable'' (e.g., by a graph acceptor \cite{ThoPOMIV96}).



\bibliography{lit}

\begin{thebibliography}{10}

\bibitem{Alur2009}
R.~Alur and P.~Madhusudan.
\newblock Adding nesting structure to words.
\newblock {\em Journal of the ACM}, 56(3):1--43, 2009.

\bibitem{Bedon15}
N.~Bedon.
\newblock Logic and branching automata.
\newblock {\em Logical Methods in Computer Science}, 11(4), 2015.

\bibitem{Bjorklund10}
H.~Bj{\"o}rklund and T.~Schwentick.
\newblock On notions of regularity for data languages.
\newblock {\em Theoretical Computer Science}, 411(4-5):702--715, 2010.

\bibitem{Bojanczy06}
M.~Bojanczyk, C.~David, A.~Muscholl, T.~Schwentick, and L.~Segoufin.
\newblock Two-variable logic on data words.
\newblock {\em ACM Transactions on Computational Logic}, 12(4):27, 2011.

\bibitem{BKM-lmcs10}
B.~Bollig, D.~Kuske, and I.~Meinecke.
\newblock Propositional dynamic logic for message-passing systems.
\newblock {\em Logical Methods in Computer Science}, 6(3:16), 2010.

\bibitem{BolligJournal}
B.~Bollig and M.~Leucker.
\newblock Message-passing automata are expressively equivalent to {EMSO} logic.
\newblock {\em Theoretical Computer Science}, 358(2-3):150--172, 2006.

\bibitem{Brand1983}
D.~Brand and P.~Zafiropulo.
\newblock On communicating finite-state machines.
\newblock {\em Journal of the ACM}, 30(2), 1983.

\bibitem{Buechi:60}
J.~B{\"u}chi.
\newblock Weak second order logic and finite automata.
\newblock {\em Z. Math. Logik, Grundlag. Math.}, 5:66--62, 1960.

\bibitem{Elgot1961}
C.~C. Elgot.
\newblock Decision problems of finite automata design and related arithmetics.
\newblock {\em Transactions of the American Mathematical Society}, 98:21--52,
  1961.

\bibitem{FisL79}
M.~J. Fischer and R.~E. Ladner.
\newblock Propositional {D}ynamic {L}ogic of regular programs.
\newblock {\em Journal of Computer and System Sciences}, 18(2):194--211, 1979.

\bibitem{GKM06}
B.~Genest, D.~Kuske, and A.~Muscholl.
\newblock {A {Kleene} theorem and model checking algorithms for existentially
  bounded communicating automata}.
\newblock {\em Information and Computation}, 204(6):920--956, 2006.

\bibitem{GKM07}
B.~Genest, D.~Kuske, and A.~Muscholl.
\newblock On communicating automata with bounded channels.
\newblock {\em Fundamenta Informaticae}, 80(1-3):147--167, 2007.

\bibitem{GradelO99}
E.~Gr{\"{a}}del and M.~Otto.
\newblock On logics with two variables.
\newblock {\em Theor. Comput. Sci.}, 224(1-2):73--113, 1999.

\bibitem{GroheS03}
M.~Grohe and N.~Schweikardt.
\newblock Comparing the succinctness of monadic query languages over finite
  trees.
\newblock In {\em Proceedings of CSL'03}, volume 2803 of {\em Lecture Notes in
  Computer Science}, pages 226--240. Springer, 2003.

\bibitem{Hanf1965}
W.~Hanf.
\newblock Model-theoretic methods in the study of elementary logic.
\newblock In J.~W. Addison, L.~Henkin, and A.~Tarski, editors, {\em The Theory
  of Models}. North-Holland, Amsterdam, 1965.

\bibitem{HenriksenJournal}
J.~G. Henriksen, M.~Mukund, K.~Narayan Kumar, M.~Sohoni, and P.~S. Thiagarajan.
\newblock A theory of regular {MSC} languages.
\newblock {\em Information and Computation}, 202(1):1--38, 2005.

\bibitem{Kuske00}
D.~Kuske.
\newblock Infinite series-parallel posets: Logic and languages.
\newblock In {\em Proceedings of {ICALP}'00}, volume 1853 of {\em LNCS}, pages
  648--662. Springer, 2000.

\bibitem{Kuske01}
D.~Kuske.
\newblock {Regular sets of infinite message sequence charts}.
\newblock {\em Information and Computation}, 187:80--109, 2003.

\bibitem{ThaWri68}
J.~W. Thatcher and J.~B. Wright.
\newblock Generalized finite automata theory with an application to a decision
  problem of second-order logic.
\newblock {\em Mathematical Systems Theory}, 2(1):57--81, 1968.

\bibitem{tho90traces}
W.~Thomas.
\newblock On logical definability of trace languages.
\newblock In {\em {P}roceedings of {A}lgebraic and {S}yntactic {M}ethods in
  {C}omputer {S}cience {(ASMICS)}}, {R}eport {TUM-I9002}, {T}echnical
  {U}niversity of {M}unich, pages 172--182, 1990.

\bibitem{ThoPOMIV96}
W.~Thomas.
\newblock Elements of an automata theory over partial orders.
\newblock In {\em Proceedings of POMIV'96}, volume~29 of {\em DIMACS}. AMS,
  1996.

\bibitem{Trakhtenbrot62}
B.~A. Trakhtenbrot.
\newblock Finite automata and monadic second order logic.
\newblock {\em Siberian Math. J}, 3:103--131, 1962.
\newblock In Russian; English translation in {\sl Amer. Math. Soc. Transl.} 59,
  1966, 23--55.

\bibitem{VardiW86}
M.~Y. Vardi and P.~Wolper.
\newblock An automata-theoretic approach to automatic program verification.
\newblock In {\em Proceedings of LICS'86}, pages 332--344. {IEEE} Computer
  Society, 1986.

\bibitem{Weis2011}
P.~Weis.
\newblock {\em Expressiveness and Succinctness of First-order Logic on Finite
  Words}.
\newblock Phd thesis, University of Massachusetts Amherst, 2011.

\bibitem{Zielonka87}
W.~Zielonka.
\newblock Notes on finite asynchronous automata.
\newblock {\em R.A.I.R.O. --- Informatique {T}h{\'e}orique et Applications},
  21:99--135, 1987.

\end{thebibliography}


\appendix

\clearpage

\section{\CFMs $\A^{\bowtie}$ for ${\bowtie} \in \{=,\prel,\mrel,\prel^{-1},\mrel^{-1}\}$}\label{app:Abowtie}

For all ${\bowtie} \in \{=,\prel,\mrel,\prel^{-1},\mrel^{-1}\}$, we give a \CFM $\A^{\bowtie}$ over $\Procs$ and $\extSigma \times 2^{\Procs \times \extSigma}$ such that $L(\A^{\bowtie}) = \{(M,\type_M^{\bowtie}) \mid M \in \MSCs{\Procs}{\extSigma}\}$.

\begin{description}\itemsep=2ex
\item[Case $\A^=$:] At every event $e$, $\A^=$ will simply ``output'' the singleton set $\{\lambda(e)\}$. Formally,
$\A^= = ((\A_p)_{p \in \Procs},\Msg,\Acc)$, with $\A_p = (S_p,\init_p,\Delta_p)$, where $S_p = \{s_p\}$ (i.e., $\init_p = s_p$ is also the local initial state of $p$), $\Acc = \{(s_p)_{p \in \Procs}\}$, and $\Msg = \{\msg\}$. Finally, for all $a \in \extSigma$ and $q \in \Procs \setminus \{p\}$, $\Delta_p$ contains the following transitions:\\[-0.5ex]
\begin{itemize}\itemsep=1ex
\item $s_p \xrightarrow{a,\{(p,a)\}} s_p$

\item $s_p \xrightarrow{a,\{(p,a)\},!,\msg,q} s_p$

\item $s_p \xrightarrow{a,\{(p,a)\},?,\msg,q} s_p$
\end{itemize}

\item[Case $\A^\prel$:] At every event $e$, $\A^\prel$ will guess the label $(p,b)$ of its process-successor (if it exists).
It will then output $\{(p,b)\}$ and go into state $b$ so that, at event $f$ with $e \prel f$, it has to read a $b$.
If the guess is that there is no process-successor, the automaton will enter $\bot$. For this construction, it is convenient to assume a \emph{set} of local initial states $I_p$ for every process, i.e., $\A_p = (S_p,I_p,\Delta_p)$. We let $S_p = I_p = \{\bot\} \cup \extSigma$, $\Acc = \{(\bot)_{p \in \Procs}\}$, and $\Msg = \{\msg\}$. Finally, for all $a,b \in \extSigma$ and $q \in \Procs \setminus \{p\}$, $\Delta_p$ contains:\\[-0.5ex]
\begin{itemize}\itemsep=1ex
\item $a \xrightarrow{a,\{(p,b)\}} b$ ~and~ $a \xrightarrow{a,\emptyset} \bot$

\item $a \xrightarrow{a,\{(p,b)\},!,\msg,q} b$ ~and~ $a \xrightarrow{a,\emptyset,!,\msg,q} \bot$

\item $a \xrightarrow{a,\{(p,b)\},?,\msg,q} b$ ~and~ $a \xrightarrow{a,\emptyset,?,\msg,q} \bot$
\end{itemize}

\item[Case $\A^\mrel$:] This is even slightly simpler than the previous case. At a \emph{send} event $e$, $\A^\mrel$ will guess the label $(q,b)$ of its message-successor. It outputs $\{(q,b)\}$ and sends $b$ to the receiving process $q$. On the other hand, the type associated with an internal or receive event is $\emptyset$. Formally, we let $S_p = \{s_p\}$, $\init_p = s_p$, $\Acc = \{(s_p)_{p \in \Procs}\}$, and $\Msg = \extSigma$. Finally, for all $a,b \in \extSigma$ and $q \in \Procs \setminus \{p\}$, $\Delta_p$ contains:\\[-0.5ex]
\begin{itemize}\itemsep=1ex
\item $s_p \xrightarrow{a,\emptyset} s_p$

\item $s_p \xrightarrow{a,\{(q,b)\},!,b,q} s_p$

\item $s_p \xrightarrow{a,\emptyset,?,a,q} s_p$
\end{itemize}

\item[Case $\A^{\prel^{-1}}$:] When reading an event $e$ with label $b$, the \CFM will enter state $b$ so that, at the process-successor event, it can access the label of $e$. Moreover, there is a distinguished initial state $\bot$. Thus, $S_p = \{\bot\} \cup \extSigma$, $\init_p = \bot$, $\Acc = \prod_{p \in \Procs} S_p$, and $\Msg = \{\msg\}$. For all $a,b \in \extSigma$ and $q \in \Procs \setminus \{p\}$, $\Delta_p$ contains:\\[-0.5ex]
\begin{itemize}\itemsep=1ex
\item $\bot \xrightarrow{a,\emptyset} a$ ~and~ $b \xrightarrow{a,\{(p,b)\}} a$

\item $\bot \xrightarrow{a,\emptyset,!,\msg,q} a$ ~and~ $b \xrightarrow{a,\{(p,b)\},!,\msg,q} a$

\item $\bot \xrightarrow{a,\emptyset,?,\msg,q} a$ ~and~ $b \xrightarrow{a,\{(p,b)\},?,\msg,q} a$
\end{itemize}

\item[Case $\A^{\mrel^{-1}}$:] Similarly, at a receive event, the receiver can access the label of the corresponding send event through the message that was sent. We let $S_p = \{s_p\}$, $\init_p = s_p$, $\Acc = \{(s_p)_{p \in \Procs}\}$, and $\Msg = \extSigma$. For all $a,b \in \extSigma$ and $q \in \Procs \setminus \{p\}$, $\Delta_p$ contains:\\[-0.5ex]
\begin{itemize}\itemsep=1ex
\item $s_p \xrightarrow{a,\emptyset} s_p$

\item $s_p \xrightarrow{a,\emptyset,!,a,q} s_p$

\item $s_p \xrightarrow{b,\{(q,a)\},?,a,q} s_p$
\end{itemize}

\end{description}

\section{Proof of Lemma~\ref{lem:cfmstrictless} (Construction of $\A^{\strictless^{-1}}$)}\label{app:Astrictless}

  Let $M=(E,\prel,\mrel,\lambda)$ be an MSC.
  To simplify notation slightly, we let, for an event $e \in E$, 
  $\rpast{e} = \{f \mid f < e \}$ and $\past{e} = \{f \mid f \strictless e \}$ (which equals $\apprelp{e}{\strictless^{-1}}$).
  Moreover, for a set $E' \subseteq E$, let $\lambda(E') = \{\lambda(e) \mid e \in E'\}$.
  Note that $\pasttype{e}=\type_M^{\strictless^{-1}}(e)$.
  
  We aim at a \CFM that ``labels'' each event $e$ of an MSC over $\Procs$ and $\extSigma$ with $\pasttype{e}$.
  We first observe that it is easy to construct a \CFM
  ``computing'' the sets $\rpasttype{e}$.
  Suppose an event $e$ on process $p$ with predecessors $f \prel e$ and
  $g \mrel e$ with $g$ on process $q$ (a situation like in Figure~\ref{fig:partition}). Then, we have
  $\rpast{e} = \rpast{f} \cup \rpast{g} \cup \{f,g\}$.
  So, to compute $\rpast{e}$, process $p$ remembers
  $\rpasttype{f}$ as well as $\lambda(f)$.
  At event $g$, process $q$ sends the message
  $\lambda(\rpast{g} \cup \{g\})$.
  So process $p$ can take the union of the set stored locally and the set sent
  by process $q$. The cases where $e$ has only one or no predecessor are similar.

  To compute $\lambda(\past{e})$, it is of course not enough to simply take $\lambda(\rpast{e})$ and remove the labels of the predecessors of $e$. This is illustrated in Figure~\ref{fig:partition}: The label $(q,b)$ of event $g$ is contained in $\lambda(\past{e}) = \lambda(\apprelp{e}{\strictless^{-1}})$ (cf.\ the blue area). However, we can modify the \CFM such that, in addition to remembering $\lambda(\rpast{e})$, it counts the number of occurrences of each
  label on each process up to 2.
  That is, instead of $\lambda(\rpast{e})$,
  it will remember a set $\mu(\rpast{e})$, where
  $\mu(e) = (\lambda(e),1)$ if $e$ is the first occurrence of
  $\lambda(e)$, and $\mu(e) = (\lambda(e),2)$ otherwise.
  For all predecessors $f \prel e$ or $g \mrel e$ of $e$, we then have:
    \begin{align*}
    \lambda(f) \in \lambda(\past{e})
    ~~\text{iff}~~ (\lambda(f),2) \in \mu(\rpast{e})
    \quad\quad\text{and}\quad\quad
    \lambda(g) \in \lambda(\past{e})
    ~~\text{iff}~~ (\lambda(g),2) \in \mu(\rpast{e})\, .  
  \end{align*}
  Indeed, if $\lambda(f) \in \lambda(\past{e})$, then $f$ is the $n$-th occurrence of $\lambda(f)$ for some $n \ge 2$, so that $(\lambda(f),2) \in \mu(\rpast{e})$.
  Conversely, if $(\lambda(f),2) \in \mu(\rpast{e})$, then there is some $f' \in \rpast{e} \setminus \{f\}$ such that $\mu(f') = (\lambda(f),1)$. Therefore, $\lambda(f) \in \lambda(\past{e})$.
  
  If $\lambda(g) \in \lambda(\past{e})$, then we can find some $g' \in \past{e}$ such that $\lambda(g') = \lambda(g)$.
  If $g' < g$, then $g$ is the $n$-th occurrence of $\lambda(g)$ for some $n \ge 2$.
  If $g < g'$, then $g'$ is the $n$-th occurrence of $\lambda(g)$ for some $n \ge 2$. In both cases, we deduce $(\lambda(g),2) \in \mu(\rpast{e})$.
  Conversely, if $(\lambda(g),2) \in \mu(\rpast{e})$, then there is some $g' \in \rpast{e} \setminus \{g\}$ such that $\lambda(g') = \lambda(g)$.
  Thus, $\lambda(g) \in \lambda(\past{e})$.

\section{Lower bounds}
\label{sec:lb}

\begin{lemma}
  Assume $|P| = 1$ and $|\Sigma| = 2$. For all $n \in \mathbb{N}$,
  there exists a sentence $\varphi \in \FOt[\prel]$ of size
  $\mathcal{O}(n^2)$ such that, for all \CFMs (i.e., finite automata)
  $\A$ with $L(\A) = L(\varphi)$, $\A$ has at least $2^{2^{n}}$ states.
\end{lemma}

\begin{proof}
  We first define an encoding of subsets of $2^{\{1,\ldots,n\}}$ as binary words.
  Let $\Sigma = \{a_0,a_1\}$.
  For all $I \subseteq \{1,\ldots,n\}$, we define words
  $w_I = a_1 a_1 a_0 a_1 a_0 x_1 a_0 x_2 \ldots a_0 x_n$
  and $\overline w_I = a_1 a_1 a_0 a_0 a_0 x_1 a_0 x_2 \ldots a_0 x_n$,
  where $x_i = a_1$ if $i \in I$ and $x_i = a_0$ otherwise.
  For all $A \subseteq 2^{\{1,\ldots,n\}}$, we define a word
  $w_A = w_{I_1} \ldots w_{I_k}$ such that $A = \{I_1,\ldots,I_k\}$.
  In addition, we let $\overline w_A = \overline w_{I_1} \ldots \overline w_{I_k}$.
  
  Our sentence $\varphi$ will be such that, for all $A,B \subseteq 2^{\{1,\ldots,n\}}$,
  $w_A \overline w_B \models \varphi$ iff $A = B$.
  We first define a formula $\beta(x)$ such that in a word $w_A$ or $\overline w_A$,
  $\beta(x)$ holds precisely at the initial positions of $w_I$ or $\overline w_I$
  factors:
  \[
    \beta(x) = a_1(x) \land \exists y.\Bigl[ x \prel y \land a_1(y) \land
    \exists x.\Bigl( y \prel x \land a_0(x)\Bigr)\Bigr]\, .
  \]
  For all $i \in \{0,\ldots,n\}$, we define an $\FO[\prel]$ formula $\alpha_i(x)$
  of size $\mathcal{O}(n)$ which holds exactly at positions $e$ such that there
  exists a path $e \prel e' \prel e_0 \prel e_0' \prel \cdots \prel e_i
  \prel e'_i$ with $\lambda(e'_i) = a_1$:
  \begin{align*}
    \alpha_0(x)
    & = \exists y.\Bigl[ x \prel y \land \exists x.\Bigl( y \prel x
      \land \exists y.\bigl( x \prel y \land a_1(y)\bigr)\Bigr)\Bigr] \\[1ex]
    \alpha_1(x)
    & = \exists y.\Bigl[ x \prel y \land \exists x.\Bigl( y \prel x \land
      \exists y.\Bigl\langle x \prel y \land \exists x.\bigl[ y \prel x \land
      \exists y.\bigl( x \prel y \land a_1(y)\bigr)\bigr]\Bigr\rangle\Bigr)\Bigr]
  \end{align*}
  and similarly for all $i \in \{0,\ldots,n\}$.

  We then let
  \begin{align*}
    \varphi = {} & \forall x. \Bigl[\bigl(\beta(x) \land \alpha_0(x)\bigr)
                   \implies \exists y. \Bigl(\beta(y) \land \lnot \alpha_0(y)
                   \land \bigwedge_{1 \le i \le n}
                   \alpha_i(x) \iff \alpha_i(y)\Bigr)\Bigr] \\
    {} \land {} &  \forall x. \Bigl[\bigl(\beta(x) \land \lnot \alpha_0(x)\bigr)
                   \implies \exists y. \Bigl(\beta(y) \land \alpha_0(y)
                  \land \bigwedge_{1 \le i \le n}
                  \alpha_i(x) \iff \alpha_i(y)\Bigr)\Bigr] \, .
  \end{align*}
  
  Assume that there exists an automaton $\A$ with less than $2^{2^n}$
  states such that $L(\varphi) = L(\A)$.
  For all $A \subseteq 2^{\{1,\ldots,n\}}$, we have $w_A\overline w_A \models \varphi$,
  hence $\A$ has some accepting run $\rho_A$ on $w_A\overline w_A$.
  Then, there exist $A,B \subseteq 2^{\{1,\ldots,n\}}$ such that $A \neq B$
  and the state of $\A$ in $\rho_A$ after reading $w_A$ is the same as the
  state of $\A$ in $\rho_B$ after reading $w_B$.
  So we can construct an accepting run of $\A$ over $w_A \overline w_B$, even
  though $w_A \overline w_B \not \models \varphi$.
\end{proof}

Contrary to the case of words and finite automata, the lower bound also
holds for sentences in $\FO[\le]$ (i.e. without the successor relations),
if $\Sigma$ is not fixed, even for $|P| = 2$.
Moreover, this second lower bound matches more precisely the upper bound.

\begin{lemma}
  \label{lem:lb1}
  Assume $|P| = 2$ and $|\Sigma| = n$ for some $n \ge 2$.  There exists a
  sentence $\varphi \in \FOt[\le]$ of size $\mathcal{O}(n)$ such that if
  $\A=((\A_p)_{p \in \Procs},\Msg,\Acc)$ is a \CFM with $L(\A)=L(\varphi)$ then
  $\A_p$ has at least $2^{2^{n-1}}$ states for some $p\in P$.
\end{lemma}

\begin{proof}
  Let $a_1, \ldots, a_{n}$ be the elements of $\Sigma$, and $p,q$ the two
  processes.
  For all $1 \le i \le n$, we let
  \[
    \alpha_i(x) = \exists y. \bigl(a_i(y) \land x \conc y\bigr) \, .
  \]
  Define
  \begin{align*}
    \varphi = {} & \forall x. \Bigl[\bigl(a_1(x) \land p(x)\bigr)
                   \implies \exists y. \Bigl(\lnot a_1(y) \land p(y) \land
                   \bigwedge_{1 \le i \le n} \alpha_i(x) \iff \alpha_i(y)\Bigr)\Bigr] \\
    {} \land {} & \forall x. \Bigl[\bigl(\lnot a_1(x) \land p(x)\bigr)
                  \implies \exists y. \Bigl(a_1(y) \land p(y) \land
                  \bigwedge_{1 \le i \le n} \alpha_i(x) \iff \alpha_i(y)\Bigl)\Bigr] \, .
  \end{align*}
  Assume that there exists a \CFM $\A=((\A_p,\A_q),\Msg,\Acc)$ such that
  $L(\A)=L(\varphi)$ and both $\A_p$ and $\A_q$ have less than
  $2^{2^{|\Sigma|-1}}$ states.

  \begin{figure}
  \centering
    \begin{tikzpicture}[xscale=0.6,yscale=1.5]
      \fill[purple!20,rounded corners=0.2cm] (0.7,1.4) rectangle (11.3,-0.4);
      \fill[teal!20,rounded corners=0.2cm] (11.7,1.4) rectangle (22.3,-0.4);
      \node[purple] at (5.5,1.6) {$M$};
      \node[teal] at (16.5,1.6) {$\overline M$};
      \foreach \i in {1,...,6,8,9,10,11} {
        \node[dot] (f\i) at (\i,1) {};
        \node[dot] (bf\i) at (11+\i,1) {}; }
      \foreach \i in {1,5,6,8,9,11} {
        \node[dot,label=below:{$a_1$}] (e\i) at (\i,0) {};
        \node[dot,label=below:{$a_2$}] (be\i) at (11+\i,0) {};
      }
      \draw (e1) -- (be11);
      \draw (f1) -- (bf11);
      \foreach \i in {1,6,9} {
        \draw[->] (f\i) -- (e\i);
        \draw[->] (bf\i) -- (be\i);
      }
      \foreach \i in {5,8,11} {
        \draw[->] (e\i) -- (f\i);
        \draw[->] (be\i) -- (bf\i);
      }
      \node[above=0cm of f2] {$a_1$};
      \node[above=0cm of f3] {$a_5$};
      \node[above=0cm of f4] {$a_2$};
      \node[above=0cm of f10] {$a_3$};
      \node[above=0cm of bf2] {$a_1$};
      \node[above=0cm of bf3] {$a_5$};
      \node[above=0cm of bf4] {$a_2$};
      \node[above=0cm of bf10] {$a_3$};
      
      \node at (0,0) {$p$};
      \node at (0,1) {$q$};
    \end{tikzpicture}
    \caption{An MSC $M \overline M$ with $M \in L$ \label{fig:lb1}}
  \end{figure}
  
  Consider the set $L$ of MSCs $M = (E,\prel,\mrel,\labloc)$ of the form
  described in Figure~\ref{fig:lb1}:
  $E_p = \{e_1,\ldots,e_k\} \cup \{f_1,\ldots,f_k\}$
  with $e_1 \prel f_1 \prel e_2 \prel \cdots \prel f_k$,
  $E_q = \bigcup_{1 = i}^k [e'_i,f'_i]$ with $f'_i \prel e'_{i+1}$,
  ${\mrel} = \{(e'_i,e_i) \mid 1 \le i \le k\} \cup
  \{(f_i,f'_i) \mid 1 \le i \le k\}$,
  and for all $e \in E_p$, $\labloc(e) = (p,a_1)$.
  For all $e \in E_p$, we define $\mu(e) = \{i \mid M,e \models \alpha_i(x)\}$,
  and $\mu(M)$ as $\{\mu(e) \mid e \in E_p\}$.
  Note that, for all $A \subseteq 2^{\{1,\ldots,n\}}$, there exists $M_A \in L$
  such that $\mu(M_A) = A$.
  
  For all $M \in L$, we let $\overline M = (E,\prel,\mrel,\labloc')$ where
  $\labloc'(e) = (p,a_2)$ for $e \in E_p$ and
  $\labloc'(e) = \labloc(e)$ for $e \in E_q$,
  and $\overline L = \{\overline M \mid M \in L\}$.
  For $M, M' \in L \cup \bar L$, we denote by $MM'$ the concatenation
  of $M$ and $M'$, defined as expected.

  For all $M \in L$, we have $M \overline M \in L(\varphi) = L(\A)$.
  In particular, for all $A \subseteq 2^{\{1,\ldots,n\}}$,
  $M_A \overline {M_A} \in L(\A)$.
  Since $\A_p$ and $\A_q$ both have less than $2^{2^{n-1}}$ states, there
  exist $A,B \subseteq 2^{\{1,\ldots,n\}}$ such that $A \neq B$ and
  there exist accepting runs $\rho_A$ and $\rho_B$ of $\A$ over 
  $M_A \overline{M_A}$ and $M_B \overline{M_B}$ such that the global state
  of $\A$ in $\rho_A$ after reading $M_A$ is the same as the global state of
  $\A$ in $\rho_B$ after reading $M_B$.
  We can then construct an accepting run of $\A$ over $M_A \overline{M_B}$,
  but $M_A \overline{M_B} \notin L(\varphi)$.
\end{proof}


\end{document}